\documentclass[conference]{IEEEtran}

\IEEEoverridecommandlockouts                              
\usepackage[utf8]{inputenc}
\usepackage{textcomp}

\usepackage{graphicx}
\usepackage{amsmath}
\usepackage[version=4]{mhchem}
\usepackage{siunitx}
\usepackage{graphicx}
\usepackage{longtable,tabularx}
\usepackage{algorithm}
\usepackage{algpseudocode}
\usepackage{tikz}

\usepackage{enumitem}

\usetikzlibrary{positioning, arrows.meta}

\usepackage{amsmath} 
\usepackage{amsthm}

\ifdefined\begintheorem\else
    
\fi
\newtheorem{proposition}{Proposition}[section]
\newtheorem{definition}{Definition}[section]

\usepackage{float} 

\usepackage{tikz}
\usetikzlibrary{shapes,arrows,calc,positioning}
\tikzstyle{bigblock} = [draw, fill=blue!20, rectangle, 
    minimum height=6em, minimum width=8em]
\tikzstyle{medblock} = [draw, fill=blue!20, rectangle, 
    minimum height=4em, minimum width=4em]    
\tikzstyle{mux} = [draw, fill=black!20, rectangle, 
    minimum height=5em, minimum width=0.1em]    
\tikzstyle{smallblock} = [draw, fill=blue!20, rectangle, 
    minimum height=2em, minimum width=3em]
    
\tikzstyle{data_block} = [draw, fill=green!20, rectangle, 
    minimum height=2em, minimum width=3em]
\tikzstyle{ops_block} = [draw, fill=blue!20, rectangle, 
    minimum height=2em, minimum width=3em]    
\tikzstyle{est_block} = [draw, fill=red!20, rectangle, 
    minimum height=2em, minimum width=3em]    
    
\tikzstyle{sum} = [draw, fill=blue!20, circle, node distance=1cm,minimum height=0.5cm]
\tikzstyle{signal} = [coordinate]
\tikzstyle{pinstyle} = [pin edge={to-,thin,black}]
\tikzstyle{block} = [draw, fill=blue!20, rectangle, 
    minimum height=3em, minimum width=9em]
\tikzstyle{blockS} = [draw, fill=blue!20, rectangle, 
    minimum height=3em, minimum width=4em]    
\tikzstyle{input} = [coordinate]
\tikzstyle{output} = [coordinate]
\usetikzlibrary{matrix}

\usetikzlibrary{positioning}
\usetikzlibrary{math} 

\usetikzlibrary{patterns}
\usetikzlibrary{calc,patterns,decorations.pathmorphing,decorations.markings}
\usepackage{hyperref}
\usepackage{xcolor}
\hypersetup{
    colorlinks,
    linkcolor={blue!100!black},
    citecolor={blue!50!black},
    urlcolor={blue!80!black}
}

\newcommand{\bc}{\begin{center}}
\newcommand{\ec}{\end{center}}
\newcommand{\benum}{\begin{enumerate}}
\newcommand{\eenum}{\end{enumerate}}
\newcommand{\nn}{\nonumber}
\newcommand{\matl}{\left[ \begin{array}}
\newcommand{\matr}{\end{array} \right]}

\renewcommand{\matl}{\begin{bmatrix}}
\renewcommand{\matr}{\end{bmatrix}}

\newcommand{\matls}{\left[ \begin{smallmatrix}}
\newcommand{\matrs}{\end{smallmatrix} \right]}
\newcommand{\isdef}{\stackrel{\triangle}{=}}

\newcommand{\inv}{^{-1}}

\newcommand{\dpder}[2]{\displaystyle\frac{\partial {#1}}{\partial {#2}}}

\newcommand{\tr}{{\rm tr}\,}

\newcommand{\rmF}{{\rm F}}

\newcommand{\rmT}{{\rm T}}

\newcommand{\rma}{{\rm a}}

\newcommand{\rmc}{{\rm c}}

\newcommand{\rmr}{{\rm r}}
\newcommand{\rms}{{\rm s}}

\newcommand{\BBR}{{\mathbb R}}

\newcommand{\SM}{{\mathcal M}}
\newcommand{\SN}{{\mathcal N}}

\newcommand{\SP}{{\mathcal P}}

\newcommand{\SX}{{\mathcal X}}
\newcommand{\SY}{{\mathcal Y}}

\newcommand{\neweqline}{\ensuremath{\nn \\ &\quad }}

\usepackage{enumitem,amssymb}
\newlist{todolist}{itemize}{2}
\setlist[todolist]{label=$\square$}
\usepackage{pifont}

\title{Dynamic Mode Decomposition-based Control of Nonlinear Systems}

\title{Matrix RLS-based Dynamic Mode Adaptive Control}
\title{Model-free Dynamic Mode Adaptive Control \\ using Matrix RLS}

\author{
Parham Oveissi
and
Ankit Goel
\thanks{Parham Oveissi is a graduate student in the Department of Mechanical Engineering, University of Maryland, Baltimore County, 1000 Hilltop Circle, Baltimore, MD 21250. {\tt\small parhamo1@umbc.edu}}%
\thanks{Ankit Goel is an Assistant Professor in the Department of Mechanical Engineering, University of Maryland, Baltimore County,1000 Hilltop Circle, Baltimore, MD 21250. {\tt\small ankgoel@umbc.edu }}%
}

\begin{document}

\maketitle


\begin{abstract}                

This paper presents a novel, model-free, data-driven control synthesis technique known as dynamic mode adaptive control (DMAC) for synthesizing controllers for complex systems whose mathematical models are not suitable for classical control design. 
DMAC consists of a dynamics approximation module and a controller module.
The dynamics approximation module is motivated by data-driven reduced-order modeling techniques and directly approximates the system's dynamics in state-space form using a matrix version of the recursive least squares algorithm. 
The controller module includes an output tracking controller that utilizes sparse measurements from the system to generate the control signal. 
The DMAC controller design technique is demonstrated through various dynamic systems commonly found in engineering applications. 
A systematic sensitivity study demonstrates the robustness of DMAC with respect to its own hyperparameters and the system's parameters. 
\end{abstract}

\section{INTRODUCTION}
Complex dynamic systems like combustion processes, fluid flows, and high-dimensional structural systems lack control-oriented mathematical models. 
While high-fidelity computational models of such systems provide highly accurate predictions, they are often too large and unwieldy for direct control synthesis. 
This challenge has led to the development of data-driven and model-free control approaches that do not require an explicit model of the system for control synthesis but rather use measurements from the system to design the controller.
Several adaptive and model-free control methods have been developed recently, including Retrospective Cost Adaptive Control (RCAC) \cite{rahman2016tutorial, oveissi2023learning}, Predictive Cost Adaptive Control (PCAC) \cite{vander2025predictive}, Model Reference Adaptive Control (MRAC), and Reinforcement Learning (RL) \cite{chen2022robust}. 
These methods rely on various optimization techniques to adjust control parameters online. 
RCAC, for example, employs a filtered retrospective cost approach that requires careful tuning of the filter parameters, whereas PCAC integrates predictive cost optimization but demands real-time constrained optimization.
Neural networks have also been explored in dynamic modeling and control applications, but their training typically requires extensive data and computational resources \cite{oveissi2024neural}.

Data-driven modeling techniques have gained significant attention recently due to their ability to extract governing equations from data without explicit first-principle modeling.
Dynamic Mode Decomposition (DMD), initially developed by the fluid dynamics community \cite{tu2013dynamic}, has been widely used in reduced-order modeling \cite{brunton2016discovering}.
Extensions such as Online DMD \cite{zhang2019online} allow adaptation to time-varying dynamics, which is critical for real-world applications where system properties change over time. 
However, while DMD provides a framework for discovering dominant system behavior, it does not inherently include control synthesis.

This work introduces a novel model-free and data-driven control synthesis technique called Dynamic Mode Adaptive Control (DMAC). 
This approach is motivated by the Dynamic Mode Decomposition (DMD) technique.
Unlike traditional model-based control methods, DMAC directly approximates system dynamics with a low-order model using sparse measurements and iteratively refines the control law, making it particularly well-suited for time-varying, high-dimensional, and nonlinear systems.
Since DMD is computationally expensive, a recursive algorithm based on matrix RLS is developed to reduce the computational cost of the algorithm.

\begin{figure*}[t]
    \centering
    \resizebox{2\columnwidth}{!}
    {
    \begin{tikzpicture}[auto, node distance=2cm,>=latex',text centered, line width = 1.5]

        \draw[draw=black, fill=yellow!10] (-2,-2.25)              
             rectangle ++(3.75,3.5) node [xshift=-5.0em, yshift=-1em] {\textbf{DMAC}} ;

        \draw[draw=black, fill=green!10] (2.5,-1.25)              
             rectangle ++(4.75,2.5) node [xshift=-6.5em, yshift=-6em] {\textbf{Sampled-data System}} ;
             
        \node [smallblock, blue, fill = blue!20, minimum width=6em, minimum height=3em] (Controller) {Control};

        \node [smallblock, right = 5 em of Controller] (zoh) {ZOH};

        \node [smallblock, red, fill=red!20,right = 2 em of zoh, minimum height=3em] (Plant) {$ \SM$};
        
        \node[circle,draw=black, fill=white, inner sep=0pt,minimum size=3pt] (rc11) at ([xshift=5em,yshift=1em]Plant) {};
        \node[circle,draw=black, fill=white, inner sep=0pt,minimum size=3pt] (rc21) at ([xshift=4em,yshift=1em]Plant) {};
        \draw [-] (rc21.north east) --node[below,yshift=.55cm]{$T_\rms$} ([xshift=.3cm,yshift=.15cm]rc21.north east) {};

        \node[circle,draw=black, fill=white, inner sep=0pt,minimum size=3pt] (rc11_xi) at ([xshift=5em,yshift=-1em]Plant) {};
        \node[circle,draw=black, fill=white, inner sep=0pt,minimum size=3pt] (rc21_xi) at ([xshift=4em,yshift=-1em]Plant) {};
        \draw [-] (rc21_xi.north east) --
        ([xshift=.3cm,yshift=.15cm]rc21_xi.north east) {};
        
        \node [smallblock, blue, fill = blue!20, below = 2 em of Controller, minimum width=6em] (DMA) {DMA};

        \draw[<-] (Controller.160) -- +(-2,0) node[xshift = 1em, yshift = 0.75em]{$r_k$};
        \draw[->] (rc11) -- +(2,0) node[xshift = -1em, yshift = 0.75em]{$y_k$};        
        \draw[->] (Controller) node[xshift = 6em, yshift = 0.75em]{$u_k$} -- (zoh);
        \draw[->] (zoh) node[xshift = 2.6em, yshift = 0.75em]{$u(t)$} -- (Plant);
        \draw[-] (Plant.32) node[xshift = 1em, yshift = 0.75em]{$y(t)$} -- (rc21) ;
        \draw[-] (Plant.-32) node[xshift = 1em, yshift = 0.75em]{$\xi(t)$} -- (rc21_xi) ;

        \draw[->] (rc11) -| +(1,-2.5) |- (-2.5,-2.5) |-(Controller.180);
        \draw[->,blue] (rc11_xi) node[xshift = 2em, yshift = 0.75em]{$\xi_k$} -| +(0.75,-1.5) |-(DMA.-10);
        
        \draw[blue,->] (Controller.0) -| +(0.4,-1) node[xshift = -0.75em, yshift = 0.1em]{$u_k$} |- (DMA.10);
        \draw[blue,->] (DMA.180) node[xshift = 0.25em, yshift = 2em]{$A_k, B_k$} -| +(-0.5,1)  |- (Controller.200);

    \end{tikzpicture}
    }
        \caption{Dynamic Mode Adaptive Control (DMAC) architecture for model-free, data-driven, and learning-based control of sampled-data systems.        
        }
        \label{fig:DMAC_architecture}
    \end{figure*}
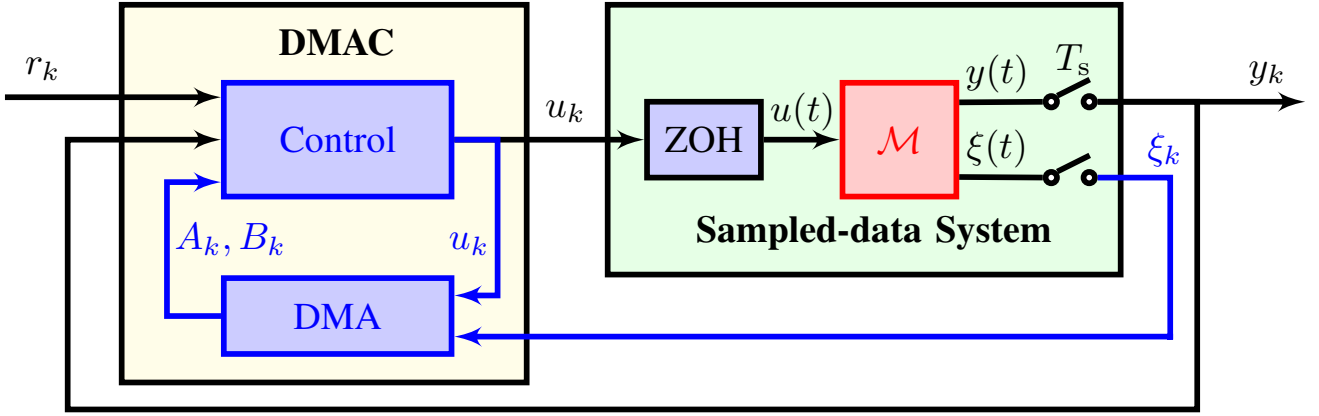

Although several model-free and data-driven control methods have been developed, they exhibit limitations in terms of computational complexity, convergence guarantees, and sensitivity to tuning parameters and the parameters of the physical system.
Unlike indirect adaptive control methods that estimate system parameters before control synthesis, DMAC approximates a low-order approximation of the system's dynamics using a matrix RLS formulation with a forgetting factor.
The forgetting factor in the RLS formulation ensures that the algorithm remains responsive to changing system dynamics, preventing the algorithm from \textit{sleeping}.
%
Moreover, while standard DMD techniques primarily focus on maintaining state estimate accuracy with reduced computational requirements, DMAC's main objective is control synthesis.
Therefore, DMAC is designed to be computationally efficient, further bridging the gap between data-driven modeling and real-time control.
Additionally, DMAC achieves efficient control synthesis with minimal computational requirements online from the data obtained from the system in real-time, distinguishing it from reinforcement learning approaches that require extensive training iterations.
By integrating principles from dynamic mode decomposition with adaptive control, DMAC offers a powerful framework for real-time model-free control of complex, high-dimensional systems.


The remainder of this paper is structured as follows. 
Section \ref{sec:DMAC} describes the DMAC algorithm, including the dynamic mode approximation and control law formulation.
Section \ref{sec:exmp} presents several numerical examples demonstrating DMAC application on representative dynamic systems, including mass-spring-damper systems, the three-mass system, a Van der Pol oscillator, and the Burgers' equation.
Finally, the paper concludes with a summary of the paper and future research directions in Section \ref{sec:conclusions}.

\section{Dynamic Mode Adaptive Control}
\label{sec:DMAC}
This section presents the dynamic mode adaptive control (DMAC) algorithm. 
Consider a continuous-time dynamic system $\SM$ whose input is $u(t) \in \BBR^{l_u}$ and the output is $y(t) \in \BBR^{l_y},$ as shown in Figure \ref{fig:DMAC_architecture}.
Letting $T_\rms>0$ denote the sample time, the system's output is sampled to generate the sampled measurements $y_k \isdef y(k T_\rms).$
The continuous-time control signal $u(t)$ is generated using zero-order hold, that is, $u(t) = u_k$ for all $t \in [kT_\rms, (k+1) T_\rms),$ where $u_k$ is the discrete-time input signal. 
The objective of the DMAC controller is to generate a discrete-time input signal $u_k$ such that the sampled output $y_k$ tracks the reference signal $r_k.$

\subsection{Dynamic Mode Approximation}
\label{sec:DynApprox}
Let $\xi_k \in \BBR^{l_\xi}$ denote the measured portion of the state of $\SM.$
Note that $\xi_k$ may or may not be the entire state of the system $\SM.$
To compute the control signal $u_k$, we first approximate linear maps $A \in \BBR^{l_\xi \times l_\xi}$ and $B \in \BBR^{l_\xi \times l_u}$ such that 
\begin{align}
    \xi_{k+1} = A \xi_{k} + B u_{k},
    \label{eq:linear_approximation}
\end{align}
which can be reformulated as
\begin{align}
    \xi_{k+1} = \Theta \phi_k.
    \label{eq:linear_approximation_AxForm_1}
\end{align}
where 
\begin{align}
    \Theta &\isdef \matl A  & B \matr \in \BBR^{l_\xi \times (l_\xi + l_u)}, 
    \\
    \phi_k &\isdef \matl \xi_{k} \\ u_{k} \matr \in \BBR^{l_\xi+l_u}. 
\end{align}
A matrix $\Theta$ such that \eqref{eq:linear_approximation_AxForm_1} is satisfied may not exist. 
However, an approximation of such a matrix can be obtained by minimizing 
\begin{align}
    J_k (\Theta)
        \isdef 
            \sum_{i=0}^{k} &\lambda^{k-i} \| \xi_{k} - \Theta \phi_{k-1} \|^2_2 
            \neweqline +
            \lambda^k \tr (\Theta^\rmT R_\Theta \Theta) ,
    \label{eq:J_k_def}
\end{align}
where
$R_\Theta \in \BBR^{(l_\xi+l_u) \times (l_\xi+l_u)}$ is a positive definite regularization matrix that ensures the existence of the minimizer of \eqref{eq:J_k_def}
and 
$\lambda \in (0,1]$ is a forgetting factor. 
In nonlinear or time-varying systems, $\Theta$ approximated by minimizing \eqref{eq:J_k_def} in the case where the state varies significantly in the state space may not be able to capture the local linear behavior at the current state.
Thus, incorporating a geometric forgetting factor to prioritize recent data over older data improves the linear approximation and prevents the algorithm from becoming sluggish. 

\begin{definition}
    Consider the cost function \eqref{eq:J_k_def}.
    For all $k\geq 0,$ define the minimizer of \eqref{eq:J_k_def} as
    \begin{align}
        \Theta_k 
            \isdef 
                \min_{\Theta \in \BBR^{l_x \times (l_x + l_u)}} J_k(\Theta).
    \end{align}
\end{definition}

\begin{proposition}
    Consider the function \eqref{eq:J_k_def}.
    Then, the minimizer $\Theta_k$ satisfies
    \begin{align}
        \Theta_k
            &=
                \Theta_{k-1} 
                +
                \left(
                    \xi_{k} - \Theta_{k-1} \phi_{k-1}
                \right)
                \phi_{k-1}^\rmT \SP_k  
            , \\
        \SP_k
            &=
                \lambda \inv \SP_{k-1} 
                -
                \lambda \inv
                \SP_{k-1} \phi_{k-1}
                \gamma_k \inv 
                \phi_{k-1}^\rmT \SP_{k-1},
    \end{align}
    where, for all $k \geq 0,$ 
    $\gamma_k \isdef \lambda  +  \phi_{k-1}^\rmT \SP_{k-1} \phi_{k-1},$ and  
    $\Theta_0 = 0,$
    $\SP_0 \isdef R_\Theta\inv. $
\end{proposition}
\begin{proof}
    See Proposition \ref{prop:theta_k_recursive} in Appendix \ref{appndx:matrix_RLS}.
\end{proof}

Note that the cost function \eqref{eq:J_k_def} is a matrix extension of the cost function typically considered in engineering applications \cite{goel2020recursive}.
As shown in \cite{Mareels1986,Mareels1988,goel2020recursive}, persistency of excitation is required to ensure that 1) the estimate converges and 2) the corresponding covariance matrix $\SP_k$ remains bounded. 
To ensure the persistency of excitation, in this paper, we introduce a zero-mean white noise in the control signal to promote persistency in the regressor $\phi_k$, as discussed in Section \ref{sec:controlUpdate}.

\subsection{Connection with Dynamic Mode Decomposition}
The identification portion of DMAC, described above, is motivated by the dynamic mode decomposition technique, which is used for reduced-order modeling of complex, large-scale dynamic systems. 
This subsection shows the equivalence between the classical DMD and the identification portion of the DMAC.
Consider a discrete system
\begin{align}
    x_{k+1} = f(x_k, u_k), 
    \label{eq:SS_NL}
\end{align}
where $x_k \in \mathbb{R}^{l_x}$ is the state vector, $u_k \in \mathbb{R}^{l_u}$ is the input vector and $f: \mathbb{R}^{l_x} \rightarrow \mathbb{R}^{l_x}$ is the  system dynamics.



For $k \geq 1,$  define the \textit{state snapshot matrix}
and the \textit{input snapshot matrix }
\begin{align}
    X_{k}
        &\isdef
            \matl 
                x_{1} & 
                x_{2} & 
                \cdots & 
                x_k
            \matr
            \in \BBR^{l_x \times k}.
    \\
    U_{k}
        &\isdef
            \matl 
                u_{1} & 
                u_{2} & 
                \cdots & 
                u_k
            \matr
            \in \BBR^{l_u \times k}.
\end{align}
The objective of DMD is to find linear maps $A \in \BBR^{l_x \times l_x}$ and $B \in \BBR^{l_x \times l_u}$ such that 
\begin{align}
    X_{k+1} = A X_{k} + B U_{k} . 
    \label{eq:linear_approximation_DMD}
\end{align}
Note that \eqref{eq:linear_approximation_DMD} can be written as 
\begin{align}
    X_{k+1 } = \Theta \SX_{k}, 
    \label{eq:linear_approximation_AxForm}
\end{align}
where 
\begin{align}
    \Theta &\isdef \matl A  & B \matr \in \BBR^{l_x \times (l_x + l_u)}, 
    \\
    \SX_k &\isdef \matl X_{k} \\ U_{k} \matr \in \BBR^{(l_x+l_u) \times k}. 
\end{align}
The matrix $\Theta$ is chosen to minimize
\begin{align}
    J_{{\rm dmd},k} (\Theta)
        \isdef 
            \| X_{k+1} - \Theta \SX_{k} \|^2_\rmF + \tr (\Theta^\rmT R_\Theta \Theta) ,
    \label{eq:J_def_DMD}
\end{align}
where $\|M\|_\rmF \isdef \tr(M M^\rmT)$ is the Frobenious norm of the matrix $M$ \cite{strang2022introduction} and $R_\Theta \in \BBR^{(l_x+l_u) \times (l_x+l_u)}$ is a positive definite regularization matrix.

Since $\| X_{k+1} - \Theta \SX_{k} \|^2_\rmF = \sum_{i=0}^{k} \| x_{i+1} - \Theta \chi_i \|^2,$
where, for all $i \geq 0,$
\begin{align}
    \chi_i 
        \isdef 
            \matl x_{i } \\ u_i \matr \in \BBR^{l_x+l_u },
\end{align}
it follows that $J_k$ minimized in the DMAC algorithm to compute a linear approximation is same as the $J_{{\rm dmd}, k}$ in classical DMD.

\subsection{Control Law Update}
\label{sec:controlUpdate}
This subsection presents the algorithm to compute the control signal $u_k$ using the dynamics approximation computed in Section \ref{sec:DynApprox}.

To track the reference signal $r_k$, the DMAC algorithm uses the fullstate feedback controller with integral action, described in Appendix \ref{sec:FSFi}.
Note that the full state refers to the state $\xi_k$ and not the system state $x_k.$ 
In particular, the control law is 
\begin{align}
    u_k = K_{\xi,k} \xi_k + K_{q,k} q_k + v_k,
\end{align}
where $K_{\xi,k} \in \BBR^{l_u \times l_\xi } $ and $K_{q,k} \in \BBR^{l_u \times l_y}$ are the time-varying fullstate feedback gain and the integrator gain, computed using the technique shown in Appendix \ref{sec:FSFi} and $v_k \sim \SN(0,\sigma_v I_{l_u})$ is a zero-mean white noise signal added to the control to promote persistency in the regressor $\phi_k$ used in the dynamic mode approximation step.

\section{Numerical Examples}
\label{sec:exmp}
This section presents the application of DMAC to several dynamic systems typically encountered in engineering applications. 
In particular, we apply DMAC to a mass-damper-spring system, a connected 3-mass system, a Van der Pol oscillator, and the Burgers equation. 
Note that the connected 3-mass system is Lyapunov stable, the Van der Pol oscillator is nonlinear system that exhibits a limit cycle, and the Burgers equation a nonlinear system governed by a partial differential equation. 


\subsection{Mass-Damper-Spring System}
Consider the MCK system
\begin{align}
    m \ddot q + c \dot q + kq = u.
    \label{eq:MCK}
\end{align}
Matlab's \href{https://www.mathworks.com/help/matlab/ref/ode45.html}{ode45} routine is used to simulate \eqref{eq:MCK}.
In this example, we set $m = 1, c = 0.5, k = 2,$ and the initial condition $x(0)$ is randomly generated using MATLAB's \href{https://www.mathworks.com/help/matlab/ref/randn.html}{randn} routine.
The output of the system is the position $q.$
The objective of the DMAC controller is to ensure that the output tracks a unit step reference signal $r.$
In this example, the DMAC algorithm updates the control signal $u_k$ every $T_\rms = 0.1 $ $\rms.$
%
%
%
%
%

To apply DMAC, we assume that 
\begin{align}
    \xi_k
        =
            \matl
                q(k T_\rms) \\
                \dot q(k T_\rms)
            \matr.
\end{align}
Note that $y_k \isdef y(T_\rms k) = q(k T_\rms).$
Since $\xi_k \in \BBR^2$ and $u_k \in \BBR,$ it follows that $\Theta_k$ is a $2 \times 3$ matrix. 
In DMAC, we set $R_\Theta = 10^2 I_3,$ the forgetting factor $\lambda = 0.995$, $R_1 = I_3$  and $R_2 = 1.$

%

%
Figure \ref{fig:MCK_DMAC_FSFI} shows the closed-loop response of the MCK system \eqref{eq:MCK} with the DMAC algorithm, where
a) shows the output $y_k$ and the reference signal $r,$ b) shows the control signal $u_k,$ c) shows the absolute value of the tracking error $z_k \isdef y_k - r$ on a logarithmic scale, and d) shows the estimate matrix $\Theta_k$ computed by DMAC. Note that the output error converges geometrically to 0.

\begin{figure}[h]
    \centering
    \includegraphics[width=\columnwidth]{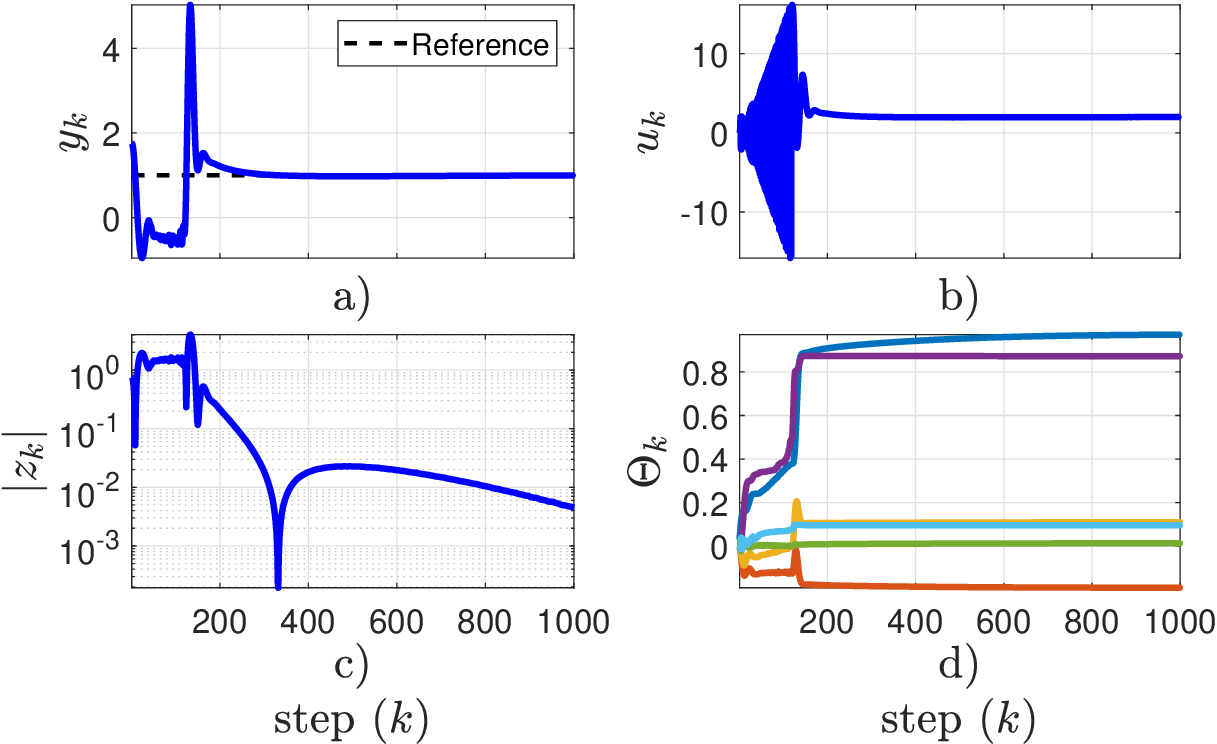}
    \caption{Closed-loop response of \eqref{eq:MCK} with DMAC. a) shows the output $y_k$ and the reference signal $r,$ b) shows the control signal $u_k,$ c) shows the absolute value of the tracking error $z_k$ on a logarithmic scale, and d) shows the estimate matrix $\Theta_k$ computed by DMAC.} 
    \label{fig:MCK_DMAC_FSFI}
\end{figure}

Next, to investigate the robustness of the DMAC algorithm to its tuning hyperparameters $\lambda, R_\Theta, R_1$ and $R_2,$  we vary each of the hyperparameters systematically by keeping other hyperparameters at their nominal values.
Figure \ref{fig:MCK_DMAC_Sensitivity_Hyperparameters} shows the effect of DMAC hyperparameters on the closed-loop response $y_k$,
where a), b), c), and d) show the effect of $\lambda, R_\Theta, R_1$ and $R_2,$ respectively.
Note that, in each case, the hyperparameter is varied by a few orders of magnitude, suggesting that DMAC is potentially robust to tuning. 

\begin{figure}[h]
    \centering
    \includegraphics[width=\columnwidth]{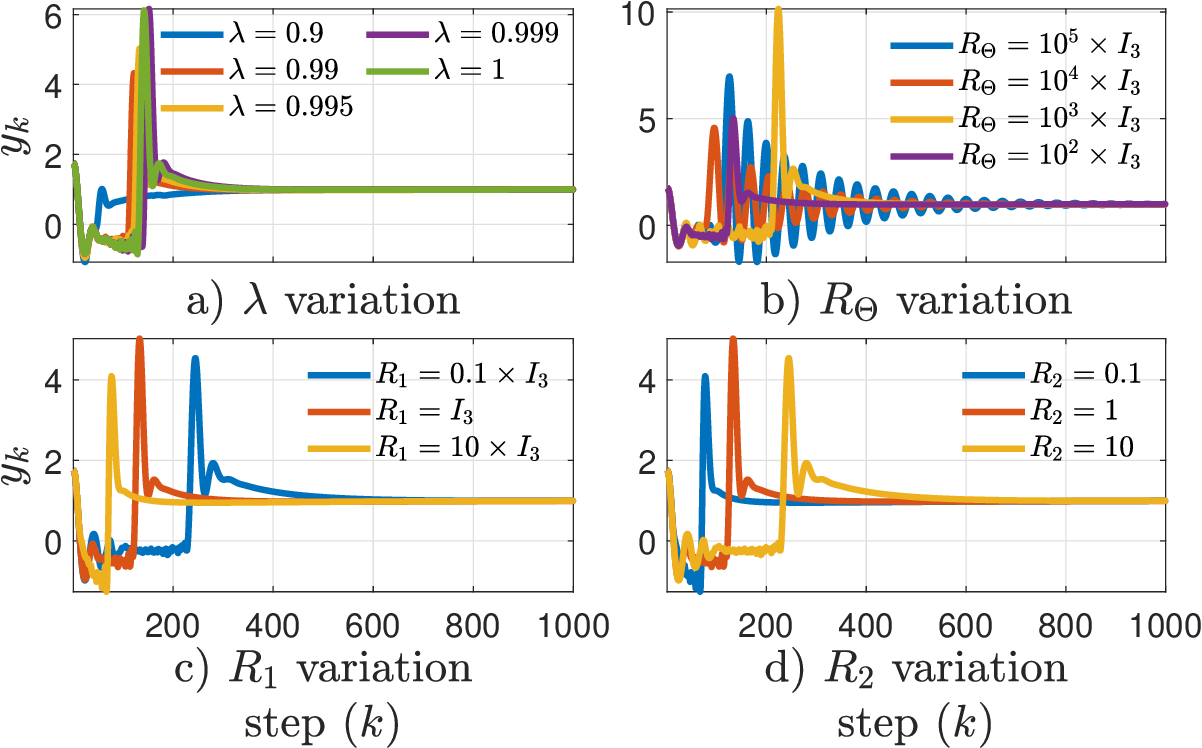}
    \caption{Effect of DMAC hyperparameters on the closed-loop performance.} 
    \label{fig:MCK_DMAC_Sensitivity_Hyperparameters}
\end{figure}

Finally, to investigate the robustness of the DMAC algorithm to the physical system, we systematically vary each of the system's parameters by keeping other parameters at their nominal values.
Figure \ref{fig:MCK_DMAC_System_Parameters} shows the effect of physical parameters on the closed-loop response $y_k$,
where a), b), and c) show the effect of $m, c,$ and $k,$ respectively.
Note that, in each case, the physical parameter is varied by a few orders of magnitude, suggesting that DMAC is potentially robust to system parameters. 


\begin{figure}[h]
    \centering
    \includegraphics[width=0.8\columnwidth]{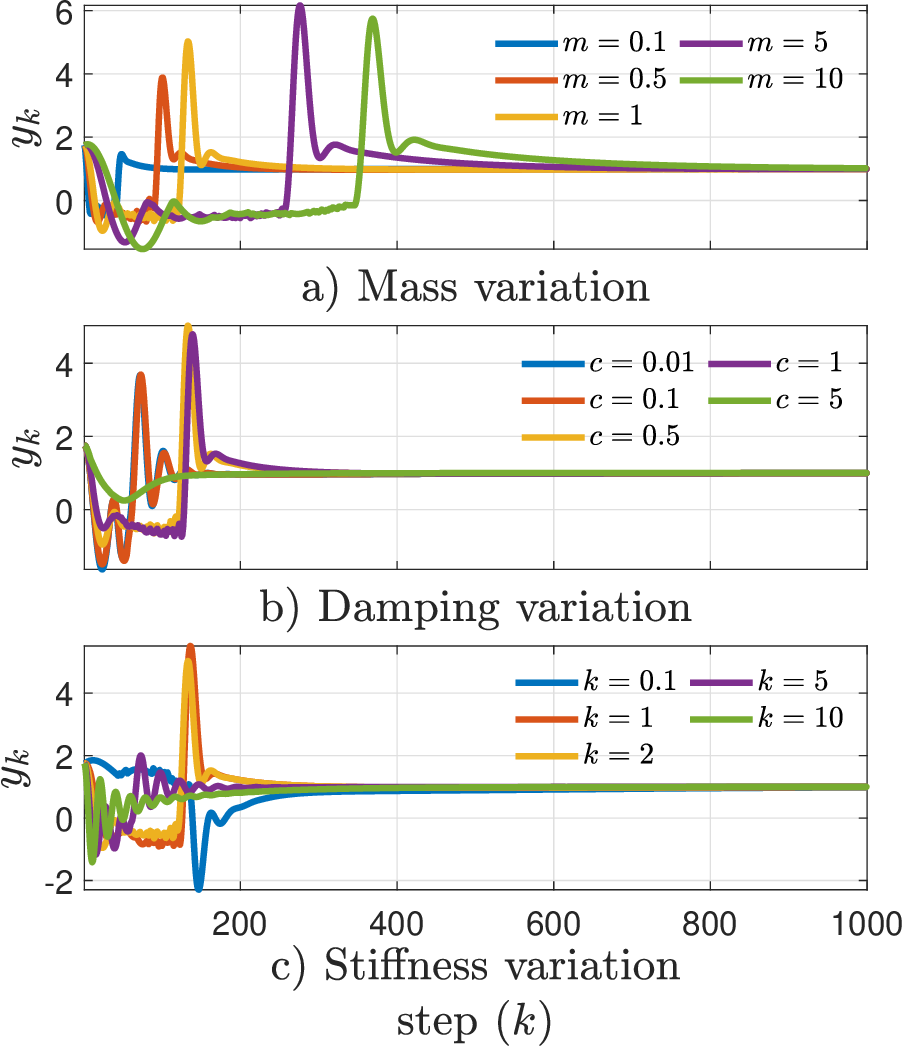}
    \caption{Effect of varying the system's physical parameters on the closed-loop performance.} 
    \label{fig:MCK_DMAC_System_Parameters}
\end{figure}






%






\subsection{3-Mass system}
Consider three masses connected with springs in a series, as shown in Figure \ref{fig:MassSpring}. 
The equations of motion are 
\begin{align}
    m \ddot{q}_1
        &= 
            -k q_1 +k (q_2-q_1) + u, \label{eq:3Meq1}
    \\
    m \ddot{q}_2
        &=
            -k (q_2-q_1)+k (q_3-q_2),
        \label{eq:3Meq2}
    \\
    m \ddot{q}_3
    &=
        -k (q_3-q_2)+k (-q_3). \label{eq:3Meq3}
\end{align}
Note that the system is Lyapunov stable since there is no damping in the system. 
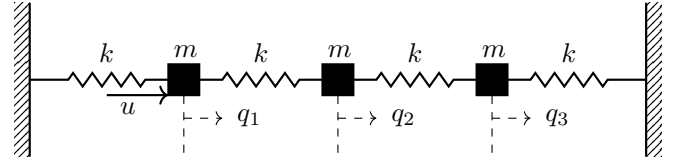
\begin{figure}[h]
    \centering
    \tikzstyle{spring}=[thick,decorate,decoration={zigzag,pre length=0.5cm,post
        length=0.5cm,segment length=.25cm}]
    \resizebox{\columnwidth}{!}
    {
    \begin{tikzpicture}
            \path [pattern=north east lines] (-0.2,1) rectangle (0,-1);
    	\draw [thick] (0, 1) -- (0,-1);
            
            \draw[spring] (0,0) -- +(2,0) node [midway,yshift=10] {$k$};
            \draw[spring] (2,0) -- +(2,0) node [midway,yshift=10] {$k$};
            \draw[spring] (4,0) -- +(2,0) node [midway,yshift=10] {$k$};
            \draw[spring] (6,0) -- +(2,0) node [midway,yshift=10] {$k$};

            \path [pattern=north east lines] (8,1) rectangle (8.2,-1);
    	\draw [thick] (8, 1) -- (8,-1);
        
            \draw [thick, fill=black] (1.8, -.2) rectangle +(.4, .4) node[xshift=-5, yshift=5] {$m$};
            \draw [thick, fill=black] (3.8, -.2) rectangle +(.4, .4) node[xshift=-5, yshift=5] {$m$};
            \draw [thick, fill=black] (5.8, -.2) rectangle +(.4, .4) node[xshift=-5, yshift=5] {$m$};
            
            \draw [dashed] (2,0) -- +(0, -1);
            \draw [dashed,->] (2,-.5) -- +(0.5, 0)
            node[xshift=10, yshift=0] {$q_1$};

            \draw [dashed] (4,0) -- +(0, -1);
            \draw [dashed,->] (4,-.5) -- +(0.5, 0)
            node[xshift=10, yshift=0] {$q_2$};

            \draw [dashed] (6,0) -- +(0, -1);
            \draw [dashed,->] (6,-.5) -- +(0.5, 0)
            node[xshift=10, yshift=0] {$q_3$};
            
            \draw [thick,->] (1,-0.2) -- +(0.8, 0)
            node[xshift=-15, yshift=-5] {$u$};
    \end{tikzpicture}
    }
    \caption{Three masses connected in series.}
    \label{fig:MassSpring}
\end{figure}
%
%
%
Matlab's \href{https://www.mathworks.com/help/matlab/ref/ode45.html}{ode45} routine is used to simulate \eqref{eq:3Meq1}-\eqref{eq:3Meq3}.
In this example, we set $m = 1, k = 2,$ and the initial condition $x(0)$ is randomly generated using MATLAB's \href{https://www.mathworks.com/help/matlab/ref/randn.html}{randn} routine.
The output of the system is the position $q_3.$
The objective of the DMAC controller is to ensure that the output tracks a unit step reference signal $r.$
In this example, the DMAC algorithm updates the control signal $u_k$ every $T_\rms = 0.1 $ $\rms.$
%
%
%
%
%

To apply DMAC, we assume that 
\begin{align}
    \xi_k
        =
            \matl
                q_1(k T_\rms) \\
                q_3(k T_\rms) \\
                \dot q_1(k T_\rms) \\
            \matr.
\end{align}
Note that $y_k \isdef y(T_\rms k) = q_3(k T_\rms).$
Since $\xi_k \in \BBR^3$ and $u_k \in \BBR,$ it follows that $\Theta_k$ is a $3 \times 4$ matrix. 
In DMAC, we set $R_\Theta = 10^2 I_4,$ the forgetting factor $\lambda = 0.999$, $R_1 = I_4$  and $R_2 = 1.$

%
%
%

Figure \ref{fig:ThreeM_DMAC_FSFI} shows the closed-loop response of the 3 mass system \eqref{eq:3Meq1}-\eqref{eq:3Meq3} with the DMAC algorithm, where
a) shows the output $y_k$ and the reference signal $r,$ b) shows the control signal $u_k,$ c) shows the absolute value of the tracking error $z_k \isdef y_k - r$ on a logarithmic scale, and d) shows the estimate matrix $\Theta_k$ computed by DMAC. Note that the output error converges geometrically to 0.

\begin{figure}[h]
    \centering
    \includegraphics[width=\columnwidth]{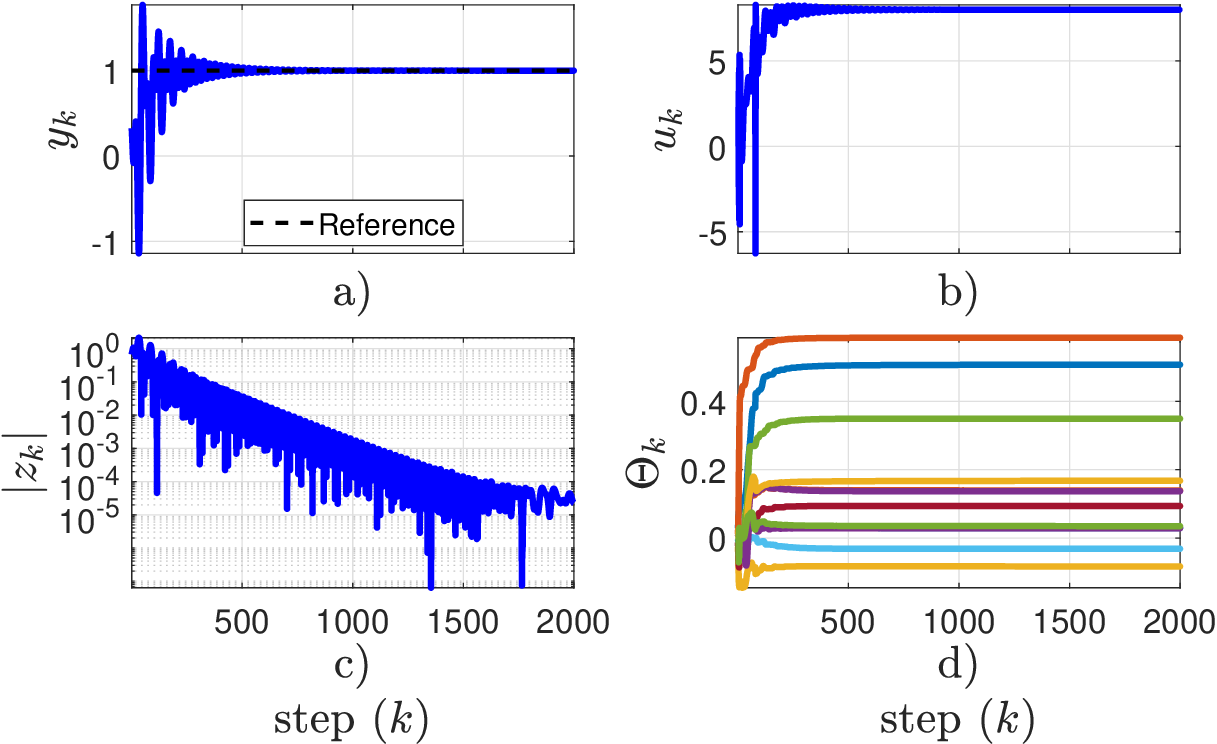}
    \caption{Closed-loop response of \eqref{eq:3Meq1}-\eqref{eq:3Meq3} with DMAC.
a) shows the output $y_k$ and the reference signal $r,$ b) shows the control signal $u_k,$ c) shows the absolute value of the tracking error $z_k$ on a logarithmic scale, and d) shows the estimate matrix $\Theta_k$ computed by DMAC.} 
    \label{fig:ThreeM_DMAC_FSFI}
\end{figure}

Next, to investigate the robustness of the DMAC algorithm to its tuning hyperparameters $\lambda, R_\Theta, R_1$ and $R_2,$  we vary each of the hyperparameters systematically by keeping other hyperparameters at their nominal values.
Figure \ref{fig:ThreeM_DMAC_Sensitivity_Hyperparameters} shows the effect of DMAC hyperparameters on the closed-loop response $y_k$,
where a), b), c), and d) show the effect of $\lambda, R_\Theta, R_1$ and $R_2,$ respectively.
Note that, in each case, the hyperparameter is varied by a few orders of magnitude, suggesting that DMAC is potentially robust to tuning.

\begin{figure}[h]
    \centering
    \includegraphics[width=\columnwidth]{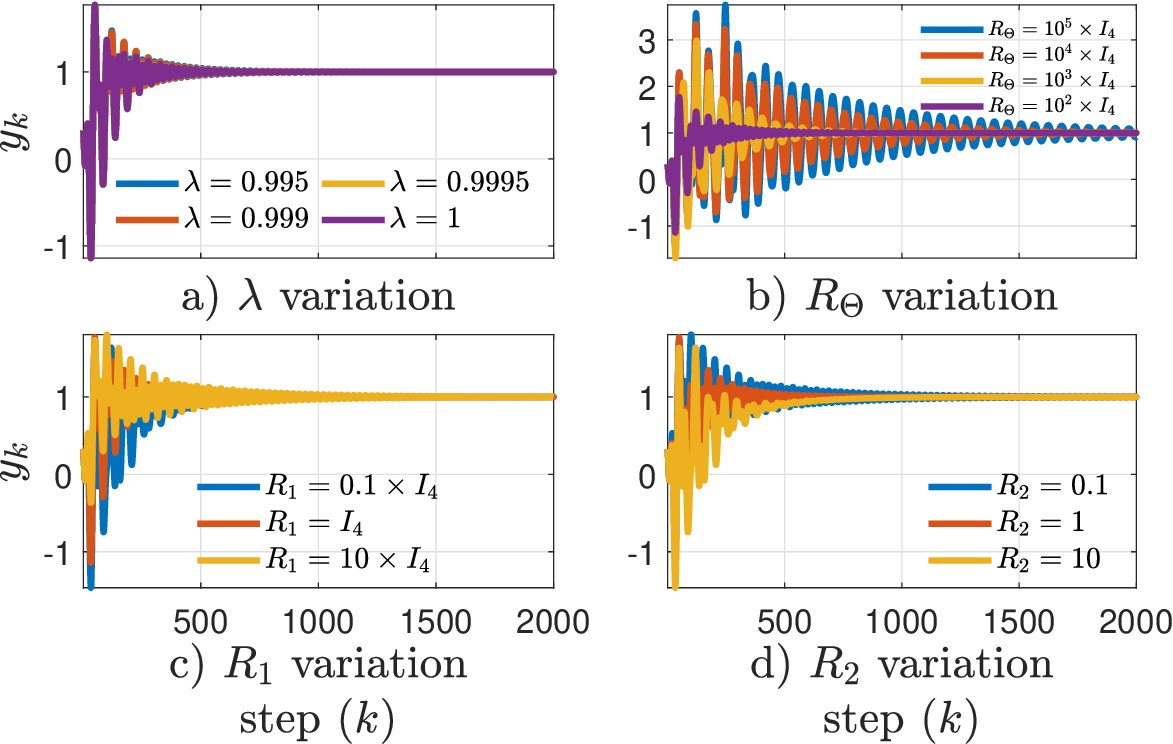}
    \caption{Effect of DMAC hyperparameters on the closed-loop performance.} 
    \label{fig:ThreeM_DMAC_Sensitivity_Hyperparameters}
\end{figure}

Finally, to investigate the robustness of the DMAC algorithm to the physical system, we systematically vary each of the system's parameters by keeping other parameters at their nominal values.
Figure \ref{fig:ThreeM_DMAC_System_Parameters} shows the effect of physical parameters on the closed-loop response $y_k$,
where a) and b) show the effect of $m$ and $k,$ respectively.
Note that, in each case, the physical parameter is varied by a few orders of magnitude, suggesting that DMAC is potentially robust to system parameters. 

\begin{figure}[h]
    \centering
    \includegraphics[width=\columnwidth]{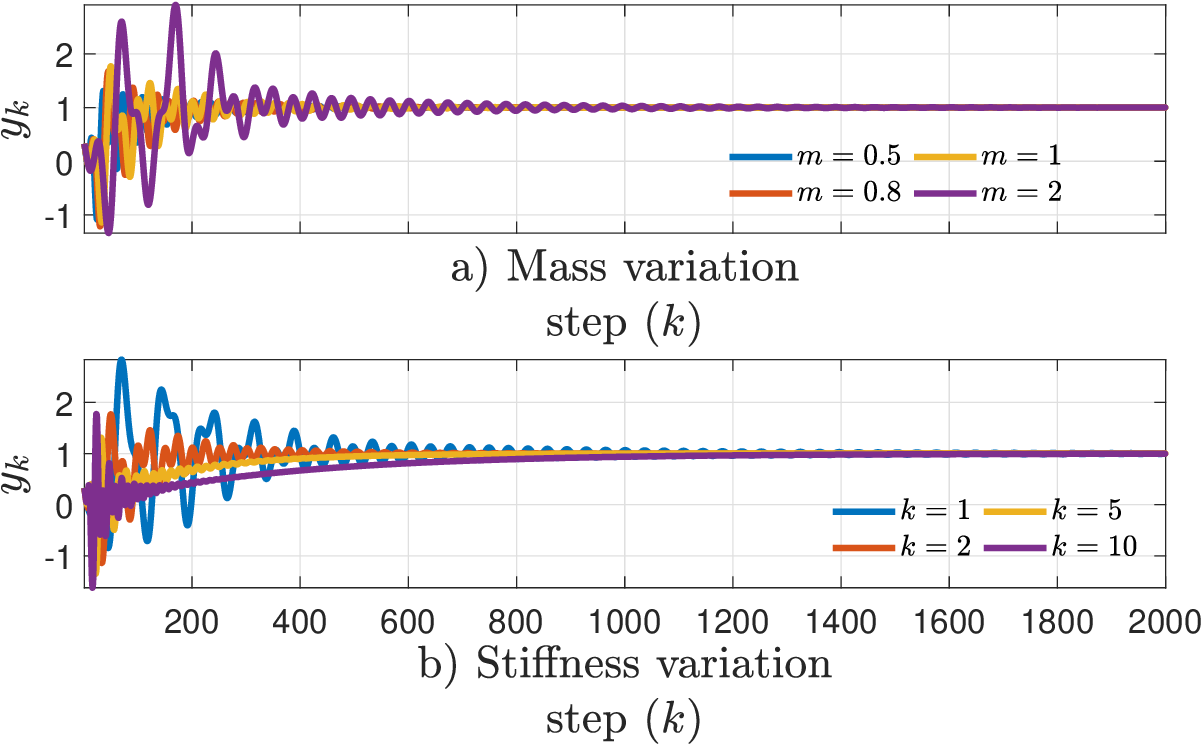}
    \caption{Effect of varying the system's physical parameters on the closed-loop performance.} 
    \label{fig:ThreeM_DMAC_System_Parameters}
\end{figure}

\subsection{Van Der Pol}
Consider the Van Der Pol oscillator
\begin{align}
    \label{eq:VDP_ODE} 
     \ddot q -\mu (1 - q^2) \dot q + q = u.
\end{align}
Matlab's \href{https://www.mathworks.com/help/matlab/ref/ode45.html}{ode45} routine is used to simulate \eqref{eq:VDP_ODE}.
In this example, we set $\mu = 1$ and the initial condition $x(0)$ is randomly generated using MATLAB's \href{https://www.mathworks.com/help/matlab/ref/randn.html}{randn} routine.
The output of the system  is the position $q.$
The objective of the DMAC controller is to ensure that the output tracks a unit step reference signal $r.$
In this example, the DMAC algorithm updates the control signal $u_k$ every $T_\rms = 0.1 $ $\rms.$

To apply DMAC, we assume that 
\begin{align}
    \xi_k
        =
            \matl
                q(k T_\rms) \\
                \dot q(k T_\rms)
            \matr.
\end{align}
Note that $y_k \isdef y(T_\rms k) = q(k T_\rms).$ 
Since $\xi_k \in \BBR^2$ and $u_k \in \BBR,$ it follows that $\Theta_k$ is a $2 \times 3$ matrix. 
In DMAC, we set $R_\Theta = 10^2 I_3,$ the forgetting factor $\lambda = 0.995$, $R_1 = I_3$  and $R_2 = 1.$

Figure \ref{fig:VDP_DMAC_FSFI} shows the closed-loop response of the VDP system \eqref{eq:VDP_ODE} with the DMAC algorithm, where
a) shows the output $y_k$ and the reference signal $r,$ b) shows the control signal $u_k,$ c) shows the absolute value of the tracking error $z_k \isdef y_k - r$ on a logarithmic scale, and d) shows the estimate matrix $\Theta_k$ computed by DMAC. Note that the output error converges geometrically to 0.
\begin{figure}[h]
    \centering
    \includegraphics[width=\columnwidth]{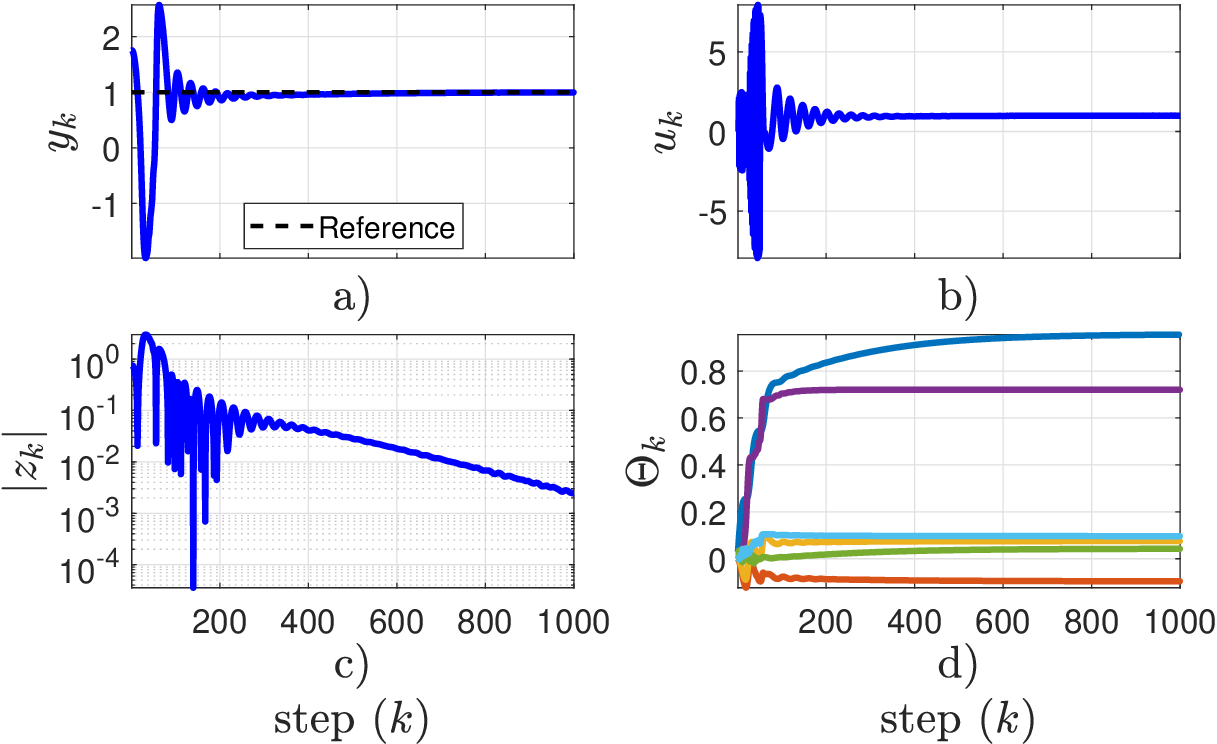}
    \caption{Closed-loop response of \eqref{eq:VDP_ODE} with DMAC. a) shows the output $y_k$ and the reference signal $r,$ b) shows the control signal $u_k,$ c) shows the absolute value of the tracking error $z_k$ on a logarithmic scale, and d) shows the estimate matrix $\Theta_k$ computed by DMAC.} 
    \label{fig:VDP_DMAC_FSFI}
\end{figure}

Next, to investigate the robustness of the DMAC algorithm to its tuning hyperparameters $\lambda, R_\Theta, R_1$ and $R_2,$  we vary each of the hyperparameters systematically by keeping other hyperparameters at their nominal values.
Figure \ref{fig:VDP_DMAC_Sensitivity_Hyperparameters} shows the effect of DMAC hyperparameters on the closed-loop response $y_k$,
where a), b), c), and d) show the effect of $\lambda, R_\Theta, R_1$ and $R_2,$ respectively.
Note that, in each case, the hyperparameter is varied by a few orders of magnitude, suggesting that DMAC is potentially robust to tuning.

\begin{figure}[h]
    \centering
    \includegraphics[width=\columnwidth]{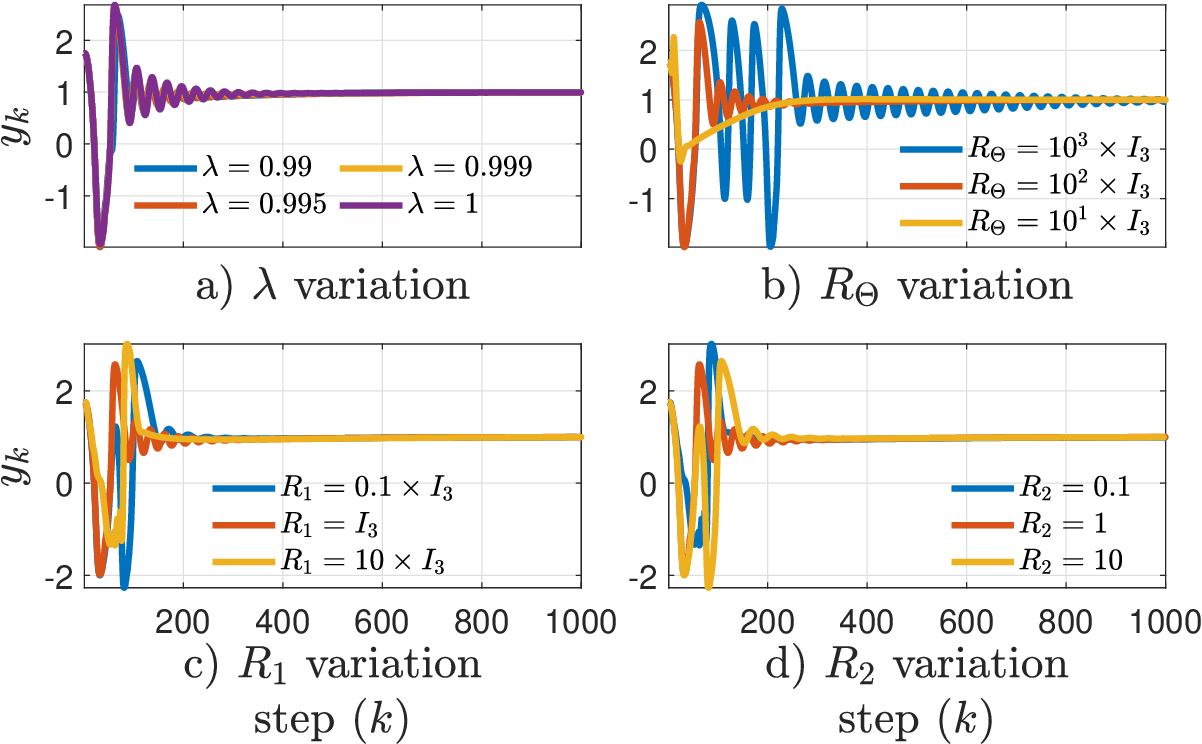}
    \caption{Effect of DMAC hyperparameters on the closed-loop performance.} 
    \label{fig:VDP_DMAC_Sensitivity_Hyperparameters}
\end{figure}

Finally, to investigate the robustness of the DMAC algorithm to the physical system, we systematically vary each of the system's parameters by keeping other parameters at their nominal values.
Figure \ref{fig:VDP_DMAC_System_Parameters} shows the effect of damping coefficient $\mu$ on the closed-loop response $y_k$,
Note that, in each case, the physical parameter is varied by a few orders of magnitude, suggesting that DMAC is potentially robust to system parameters. 
\begin{figure}[h]
    \centering
    \includegraphics[width=\columnwidth]{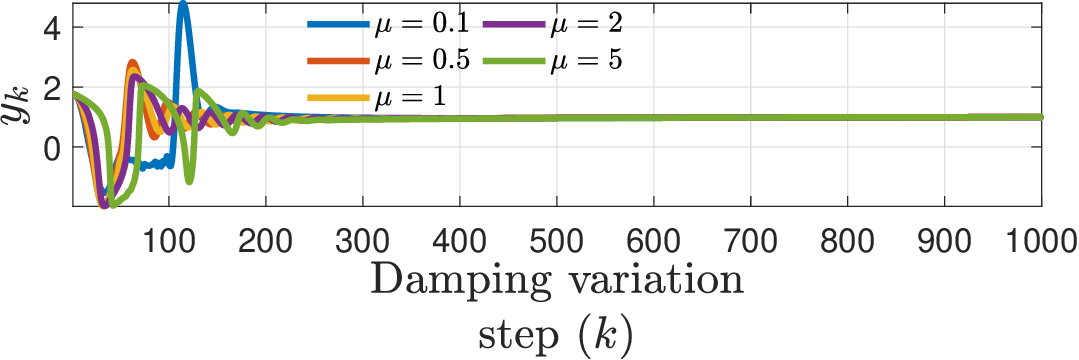}
    \caption{Effect of varying the system's physical parameters on the closed-loop performance.} 
    \label{fig:VDP_DMAC_System_Parameters}
\end{figure}

\subsection{Burgers' Equation}

Consider the one-dimensional Burgers' equation given by
\begin{align}
\label{eq:Burgers_PDE}
\frac{\partial w}{\partial t} + w \frac{\partial w}{\partial x} = \nu \frac{\partial^2 w}{\partial x^2} + u_\rmc(x,t),
\end{align}
where $w(x,t)$ may represent a pressure or a velocity field, $\nu$ is the viscosity coefficient, and $u_\rmc(x,t)$ is the control input applied at the location $x$ at time $t$.

To simulate the Burgers equation \eqref{eq:Burgers_PDE},
we discretize \eqref{eq:Burgers_PDE} in space using a uniform grid with $N$ nodes over the domain $x \in [0, 2\pi].$
Thus, the spatial step size $\Delta x = 2\pi / (N-1)$ and 
the $i$th spatial node is $x_i = (i-1) \Delta x.$
The approximation of the field variable $w(x,t)$ at the spatial node $x_i$ is denoted by $w_i(t).$
At each node, the convective and diffusive terms are then approximated using central differences, that is, 
\begin{align}
    \dpder{w}{x} \Bigg |_{x_i}
        &\approx
            \frac{w_{i+1} -w_{i-1}}{2 \Delta x} 
    \\
    \frac{\partial^2 w}{\partial x^2} \Bigg|_{x_i}
        &\approx
            \frac{w_{i+1} - 2 w_i + w_{i-1}}{\Delta x^2}.
\end{align}
%
Thus, for $i = 1,2,3, \ldots, N,$ 
\begin{align}
\dot{w}_i 
    &=
        - w_i \frac{w_{i+1} - w_{i-1}}{2 \Delta x} 
        \neweqline
        + \nu \frac{w_{i+1} - 2 w_i + w_{i-1}}{\Delta x^2} + u_{\rmc,i}(t). 
    \label{eq:discretizedODE}
\end{align}
To ensure a continuous flow, we impose periodic boundary conditions, where the function values wrap around the domain.
In particular, the periodic boundary conditions are imposed by setting  $w_{0} = w_N$ and $w_{N+1} = w_1.$

In this work, we set $N=100,$ $\nu = 0.1,$ and initialize $w(x,0)$ randomly using MATLAB's \href{https://www.mathworks.com/help/matlab/ref/randn.html}{randn} routine. 
Matlab's \href{https://www.mathworks.com/help/matlab/ref/ode45.html}{ode45} routine is used to simulate \eqref{eq:discretizedODE}.
The output of the system is assumed to be the field variable at the $61$st node, that is, $y = w_{61}.$
The objective of the DMAC controller is to ensure that the output tracks a unit step reference signal $r.$
In this example, the DMAC algorithm updates the control signal $u_k$ every $T_\rms = 0.01 $ $\rms.$

To reflect the physical scenario with limited sensors, we assume that the field variable $w_i$ is measured at only a few sparse locations. 
In particular, we assume that $w_i$ is measured at the nodes $i \in \{1, 16, 31, 46, 61, 76, 91 \}.$
Therefore, to apply DMAC, we assume that $x_k = \matl w_1 & w_{16} & w_{31} & w_{46} & w_{61} & w_{76} & w_{91} \matr^\rmT \Big |_{t=k T_\rms}.$
Finally, in this example, the control is applied at $55$th, therefore, $u = u_{\rmc,55}$



With the choice of $x_k$ and $u_k$ described above, it follows that $\Theta_k$ is a $7 \times 8$ matrix.
In DMAC, we set 
$R_\Theta = 10^2 I_8,$ 
forgetting factor of $\lambda = 0.9995,$
$R_1 = 10 \times I_8,$ and
$R_2 = 0.1.$

Figure \ref{fig:Burgers_DMAC_FSFI} shows the closed-loop response of the Burgers' system \eqref{eq:Burgers_PDE} with the DMAC algorithm, where
a) shows the output $y_k$ and the reference signal $r,$ b) shows the control signal $u_k,$ c) shows the absolute value of the tracking error $z_k \isdef y_k - r$ on a logarithmic scale, and d) shows the estimate matrix $\Theta_k$ computed by DMAC. Note that the output error converges geometrically to 0. Figure \ref{fig:Burgers_DMAC_states_contour} shows the closed-loop response of the discretized state of the Burgers equation.  
 
\begin{figure}[h]
\centering
\includegraphics[width=\columnwidth]{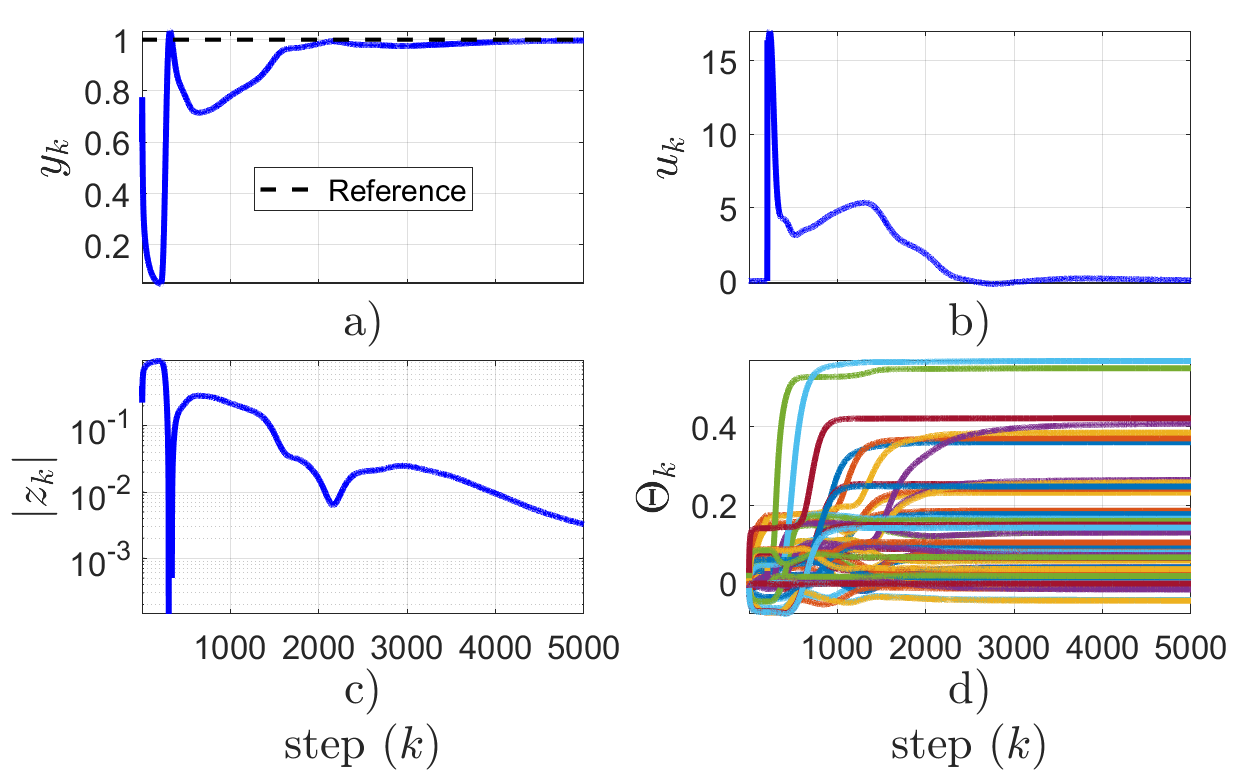}
\caption{Closed-loop response of \eqref{eq:Burgers_PDE} with DMAC. a) shows the output $y_k$ and the reference signal $r,$ b) shows the control signal $u_k,$ c) shows the absolute value of the tracking error $z_k$ on a logarithmic scale, and d) shows the estimate matrix $\Theta_k$ computed by DMAC.}
\label{fig:Burgers_DMAC_FSFI}
\end{figure}

\begin{figure}[h]
\centering
\includegraphics[width=\columnwidth]{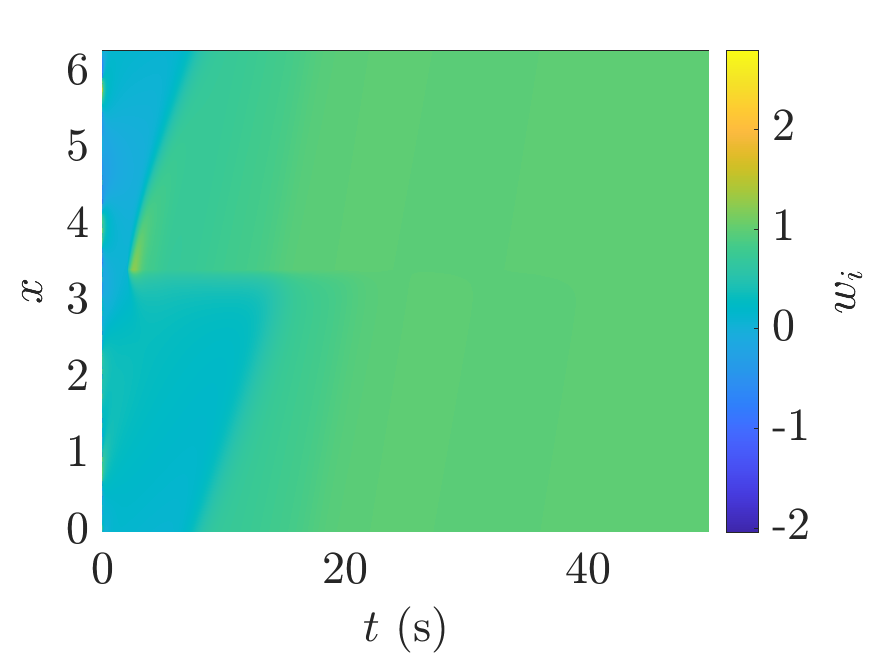}
\caption{Closed-loop response of the discretized Burgers equation.}
\label{fig:Burgers_DMAC_states_contour}
\end{figure}

Next, to investigate the robustness of the DMAC algorithm to its tuning hyperparameters $\lambda, R_\Theta, R_1$ and $R_2,$  we vary each of the hyperparameters systematically by keeping other hyperparameters at their nominal values.
Figure \ref{fig:Burgers_DMAC_Sensitivity_Hyperparameters} shows the effect of DMAC hyperparameters on the closed-loop response $y_k$,
where a), b), c), and d) show the effect of $\lambda, R_\Theta, R_1$ and $R_2,$ respectively.
Note that, in each case, the hyperparameter is varied by a few orders of magnitude, suggesting that DMAC is potentially robust to tuning.
\begin{figure}[h]
\centering
\includegraphics[width=\columnwidth]{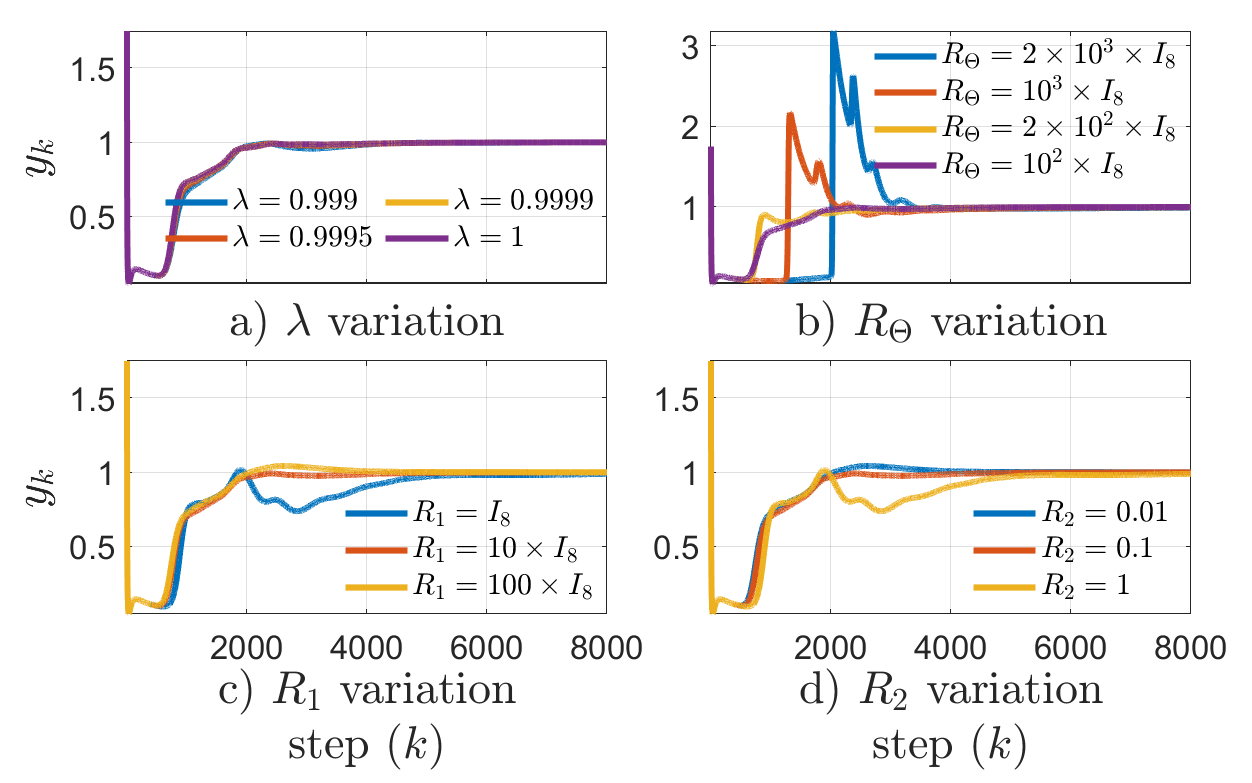}
\caption{Effect of DMAC hyperparameters on the closed-loop performance.}
\label{fig:Burgers_DMAC_Sensitivity_Hyperparameters}
\end{figure}

Finally, to investigate the robustness of the DMAC algorithm to the physical system, we systematically vary each of the system's parameters by keeping other parameters at their nominal values.
Figure \ref{fig:Burgers_DMAC_System_Parameters} shows the effect of the viscosity $\nu$ on the closed-loop response $y_k$.
Note that the physical parameter is varied by a few orders of magnitude, suggesting that DMAC is potentially robust to system parameters. 
\begin{figure}[H]
\centering
\includegraphics[width=\columnwidth]{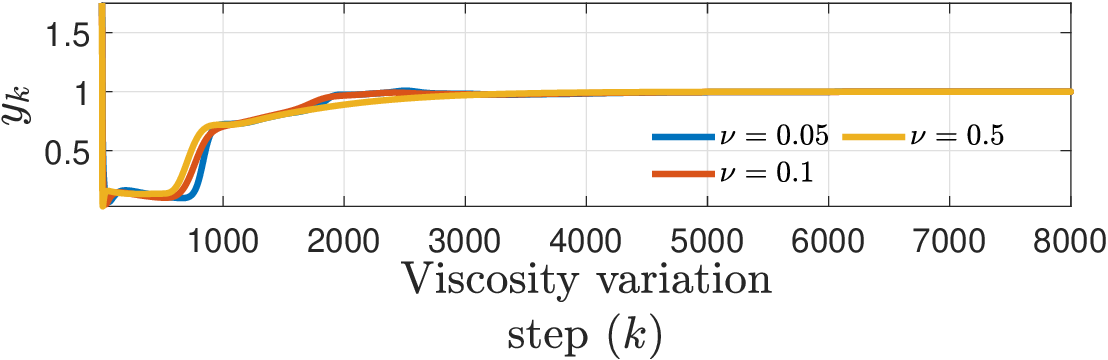}
\caption{Effect of varying the system's physical parameters on the closed-loop performance.}
\label{fig:Burgers_DMAC_System_Parameters}
\end{figure}

\section{Conclusions}
\label{sec:conclusions}

This paper introduced a data-driven adaptive control technique called dynamic mode adaptive control (DMAC). 
The DMAC algorithm consists of a low-order dynamics approximation system motivated by the dynamic mode decomposition and a reference tracking controller based on the low-order approximated dynamics. 
The technique is applied to several dynamical systems of engineering interests. 
The numerical examples demonstrate the simplicity of the controller construction and the robustness of the algorithm to its own learning hyperparameters as well as the physical parameters of the system. 
In particular, the technique is demonstrated in model-free, command-following problems in a mass-damper-spring system, a three-mass system, the Van der Pol system, and the Burgers equation. 

The future work is focused on 
\begin{enumerate}
    \item reducing the requirement of probing signals without sacrificing performance, 
    \item introducing constrained optimization to ensure controllability of the low-order approximation, and 
    \item providing theoretical guarantees in ideal scenarios. 
\end{enumerate}



\section{Appendix}
\subsection{Full-State Feedback Control with Integral Action}
\label{sec:FSFi}

Consider the system
\begin{align}
    x_{k+1} &= Ax_k + B u_k, \\
    y_k &= Cx_k.
\end{align}
To stabilize the system and track a reference command $r_k,$ define the integrator state
\begin{align}
    q_k \isdef \sum_{i=0}^k \big( r_i - y_i \big) \in \BBR^{l_y}.
\end{align}
Note that the integrator state $q_k$ satisfies
\begin{align}
    q_{k+1} = q_k + e_k,
\end{align}
where the output error $e_k \isdef r_k - y_k = r_k - Cx_k.$

Define the augmented state
\begin{align}
    x_{\rma,k} 
        =
        \matl
            x_k \\ 
            q_k
        \matr
        \in \BBR^{(l_x + l_y)}.
\end{align}
Note that the augmented state satisfies
\begin{align}
    x_{\rma,k+1} &= A_\rma x_{\rma,k} + B_\rma u_k + B_\rmr r_k, \\
    y_k &= C_\rma x_{\rma,k},
\end{align}
where
\begin{align}
    A_\rma &\isdef \matl
        A & 0 \\ 
        -C & I
    \matr, \quad B_\rma \isdef \matl
        B \\ 
        0
    \matr, \\ 
    B_\rmr &\isdef \matl
        0 \\ 
        I
    \matr, \quad 
    C_\rma \isdef \matl
        C & 0
    \matr.
\end{align}

To stabilize the system and track a reference command $r_k,$ consider the control law
\begin{align}
    u_k = K_x x_k + K_q q_k,
\end{align}
where $K_x \in \BBR^{l_u \times l_x}$ is the state feedback gain matrix and $K_q \in \BBR^{l_u \times l_y}$ is the integral gain. Note that 
\begin{align}
    u_k = K_\rma x_{\rma,k},
\end{align}
where 
$K_\rma \isdef \matl K_x & K_q \matr,$ which implies that
\begin{align}
    x_{\rma,k+1} = (A_\rma + B_\rma K_\rma) x_{\rma,k} + B_\rmr r_k.
\end{align}
If the pair $(A_\rma, B_\rma)$ is stabilizable, then there exists a $K_\rma$ such that $(A_\rma + B_\rma K_\rma)$ is Schur stable.
Furthermore, if $r_k=r$ is a constant, then $x_{\rma,k}$ converges to a constant, which implies that $q_k$ converges, which in turn implies that $e_k$ converges to zero, which finally implies that $y_k$ converges to the constant $r.$


\subsection{Matrix RLS}

\label{appndx:matrix_RLS}
Consider the cost function
\begin{align}
    J(k,\Theta)
        &=
            \sum_{i=1}^k
            \lambda^{k-i}
            (y_i - \Theta x_i )^\rmT(y_i - \Theta x_i )
            \nn \\ &\quad
            + \tr \lambda^{k} R_\Theta \Theta \Theta^\rmT,
    \label{eq:cost_def}
\end{align}
where for all $i\geq 1,$ $y_i \in \BBR^{l_y}$ and
$x_i \in \BBR^{l_x},$ 
$\Theta \in \BBR^{l_y \times l_x}$, 
$R_\Theta \in \BBR^{l_y \times l_y}$ is positive definite, and 
$\lambda \in (0,1]$ is the forgetting factor.
Note that the cost function can be written as
\begin{align}
    J(k,\Theta)
        &=
            \sum_{i=1}^k \lambda^{k-i} y_i^\rmT y_i 
            + \tr \SX_k \Theta ^\rmT \Theta   
            - 2 \tr \SY_k \Theta,
\end{align}
where
\begin{align}
    \SX_k 
        &\isdef 
            \sum_{i=1}^k  \lambda^{k-i} x_i x_i^\rmT + \lambda^{k} R_\Theta   
    , \\
    \SY_k 
        &\isdef 
            \sum_{i=1}^k  \lambda^{k-i} x_i y_i^\rmT .
\end{align}

\begin{proposition}
    \label{prop:theta_k_def}
    Consider the cost function \eqref{eq:cost_def}.
    Let $\Theta_k$ denote the minimizer of the cost function \eqref{eq:cost_def}.
    Then, 
    \begin{align}
        \Theta_k
            =
                \SY_k^\rmT \SX_k^{-1}.
    \end{align}
\end{proposition}
\begin{proof}
Using Fact \ref{fact:trace_derivatives}, it follows that 
\begin{align}
    \frac{\partial}{\partial \Theta} 
    \tr \SX_k \Theta ^\rmT \Theta 
        &=
            2 \Theta \SX_k , 
    \quad 
    \frac{\partial}{\partial \Theta} \tr\SY_k \Theta
        =
            \SY_k^\rmT, 
\end{align}
and thus
\begin{align}
     \frac{\partial}{\partial \Theta}
     J(k,\Theta)
        =
            2 \Theta \SX_k - 2\SY_k^\rmT.
\end{align}
Setting the gradient equal to zero yields the minimizer. 
\end{proof}

\begin{proposition}
    \label{prop:theta_k_recursive}
    Consider the cost function \eqref{eq:cost_def}.
    Let $\Theta_k$ denote the minimizer of the cost function \eqref{eq:cost_def}.
    Then, the minimizer $\Theta_k$ satisfies
    \begin{align}
        \Theta_k
            &=
                \Theta_{k-1} 
                +
                \left(
                    y_k - \Theta_{k-1} x_k
                \right)
                x_k^\rmT \SP_k  
            , \\
        \SP_k
            &=
                \lambda \inv \SP_{k-1} 
                -
                \lambda \inv \SP_{k-1} x_k \cdot
                \neweqline
                (\lambda  +  x_k^\rmT \SP_{k-1} x_k)\inv x_k^\rmT \SP_{k-1},
    \end{align}
    where $\SP_0 \isdef R_\Theta\inv. $
\end{proposition}

\begin{proof}
\hypertarget{proof:theta_k_rec}{Proof of Proposition \ref{prop:theta_k_recursive}.}
Note that 
\begin{align}
    \SX_k 
        &=
            \sum_{i=1}^k  \lambda^{k-i} x_i x_i^\rmT + \lambda^{k} R_\rma
        =
            \lambda \SX_{k-1} + x_k x_k^\rmT, 
    \\
    \SY_k 
        &=
            \sum_{i=1}^k  \lambda^{k-i} x_i y_i^\rmT 
        =
            \lambda \SY_{k-1} + x_k y_k^\rmT.
\end{align}

Define $\SP_k \isdef \SX_k^{-1} .$ Then, 
\begin{align}
    \SP_k
        &=
            \left( \lambda \SX_{k-1} + x_k x_k^\rmT \right)^{-1} 
        \nn \\ 
        &=
            \lambda \inv \SX_{k-1}\inv 
            -
            \lambda \inv \SX_{k-1}\inv 
            x_k 
            \gamma_k \inv 
            x_k^\rmT 
            \lambda \inv \SX_{k-1}\inv
        \nn \\ 
        &=
            \lambda \inv \SP_{k-1} 
            -
            \lambda \inv \SP_{k-1} x_k 
            \gamma_k \inv 
            x_k^\rmT \SP_{k-1},
        \\
    \Theta_k
        &=
            \SY_k^\rmT \SX_k\inv
        \nn \\
        &=
            (\lambda \SY_{k-1}^\rmT + y_k x_k^\rmT) \SX_k\inv
        \nn \\
        &=
            (\lambda \Theta_{k-1} \SX_{k-1} + y_k x_k^\rmT) \SX_k\inv
        \nn \\
        &=
            (\Theta_{k-1} (\SX_{k} - x_k x_k^\rmT) + y_k x_k^\rmT) \SX_k\inv
        \nn \\
        &=
            \Theta_{k-1} \SX_{k} \SX_k\inv 
            - \Theta_{k-1} x_k x_k^\rmT \SX_k\inv
            + y_k x_k^\rmT \SX_k\inv
        \nn \\
        &=
            \Theta_{k-1} 
            - \Theta_{k-1} x_k x_k^\rmT \SP_k
            + y_k x_k^\rmT \SP_k
        \nn \\
        &=
            \Theta_{k-1} 
            +
            \left(
                y_k - \Theta_{k-1} x_k
            \right)
            x_k^\rmT \SP_k            ,
\end{align}
which completes the proof. 
\end{proof}

\begin{proposition}
    \label{fact:trace_derivatives}
    Let $X$, $A,$ $B$ be matrices. Then,
    \begin{align}
        \frac{\partial}{\partial X} \tr B X X^\rmT 
            &=
                X (B+B^\rmT ) , 
        \\
        \frac{\partial}{\partial X} \tr A X B
            &=
                A^\rmT B. 
    \end{align}
\end{proposition}
\begin{proof}
    See \cite{petersen2008matrix}.
\end{proof}



\bibliography{DMDbib}

\end{document}